\documentclass[a4paper,11pt]{article}
\usepackage{a4wide}
\usepackage{authblk}
\usepackage{amsmath,amsthm,amssymb}
\usepackage{commath}
\usepackage{mathtools}
\usepackage[utf8]{inputenc}
\usepackage[swedish,english]{babel}
\usepackage{csquotes}
\usepackage{tikz}
\usepackage{float}
\usepackage{verbatim}
\usepackage{graphicx}
\usepackage{sidecap}
\usepackage{caption}
\usepackage{standalone}
\usepackage[pdftitle={Hes1 Modelling},
pdfauthor={Gesina Menz, Stefan Engblom},
pdffitwindow=true,
breaklinks=true,
colorlinks=true,
urlcolor=blue,
linkcolor=red,
citecolor=red,
anchorcolor=red]{hyperref}
\usepackage{orcidlink}
\usepackage[backend=biber,style=numeric,citestyle=numeric,sorting=nyt,
sortcites,giveninits=true,url=false,eprint=false]{biblatex}
\usepackage{sidenotes}

\DeclareSourcemap{
  \maps[datatype=bibtex]{
    \map{
      \step[fieldset=issn, null]
      \step[fieldset=abstract, null]
      \step[fieldset=month, null]
    }
  }
}

\newcommand{\review}[1]{#1}

\newcommand\keywords{\medskip \noindent \textbf{Keywords:}\ }
\newcommand\acmsubjects{\medskip \noindent \textbf{AMS subject classification:}\ }

\numberwithin{equation}{section}
\numberwithin{table}{section}
\numberwithin{figure}{section}

\theoremstyle{plain}
\newtheorem{theorem}{Theorem}[section]
\newtheorem{lemma}[theorem]{Lemma}
\newtheorem{proposition}[theorem]{Proposition}

\theoremstyle{definition}

\renewcommand{\figref}[1]{Fig.~\ref{fig:#1}}
\newcommand{\tabref}[1]{Tab.~\ref{tab:#1}}

\newcommand{\Din}{D_{\text{in}}}

\newcommand{\yin}{y_{\text{in}}}
\newcommand{\xin}{x_{\text{in}}}

\newcommand{\stoich}{\mathbb{N}}
\newcommand{\Intdom}{\mathbf{Z}}

\newcommand{\VoxVol}{V}

\newcommand{\Prob}{\mathbf{P}}

\title{Modelling Population-Level Hes1 Dynamics: Insights from a
  Multi-Framework Approach}

\author[1]{Gesina Menz\orcidlink{0000-0002-3589-9824}}

\author[1,2]{Stefan Engblom\thanks{\textit{Corresponding author.}  URL:
    \url{http://www.stefanengblom.org}, telephone +46-18-471 27 54,
    fax +46-18-51 19 25.}\orcidlink{0000-0002-3614-1732}}

\affil[1]{{\footnotesize Division of Scientific Computing, Department
    of Information Technology, Uppsala University, SE-751 05 Uppsala,
    Sweden. E-mail: \href{mailto:gesina.menz@it.uu.se}{gesina.menz},
    \href{mailto:stefane@it.uu.se}{stefane@it.uu.se}.}}

\affil[2]{{\footnotesize Science for Life Laboratory, Department of
    Information Technology, Uppsala University}}

\date{\today}

\begin{document}

\maketitle

\begin{abstract}
  Mathematical models of living cells have been successively refined
  with advancements in experimental techniques. A main concern is
  striking a balance between modelling power and the tractability of
  the associated mathematical analysis.

  In this work we model the dynamics for the transcription factor
  Hairy and enhancer of split-1 (Hes1), whose expression oscillates
  during neural development, and which critically enables stable fate
  decision in the embryonic brain. We design, parametrise, and analyse
  a detailed spatial model using ordinary differential equations
  (ODEs) over a grid capturing both transient oscillatory behaviour
  and fate decision on a population-level. We also investigate the
  relationship between this ODE model and a more realistic grid-based
  model involving intrinsic noise using mostly directly biologically
  motivated parameters.

  While we focus specifically on Hes1 in neural development, the
  approach of linking deterministic and stochastic grid-based models
  shows promise in modelling various biological processes taking place
  in a cell population. In this context, our work stresses the
  importance of the interpretability of complex computational models
  into a framework which is amenable to mathematical analysis.

  \keywords{fate decision, neurogenesis, cellular synchronisation,
    genetic oscillator, pattern formation.}

  \acmsubjects{\textit{Primary:} 92-10, 92B25, 92C15;
    \textit{secondary:} 34A33, 60J20, 34C60, 34F10.}



    

  \medskip
  \noindent
  \textbf{Statements and Declarations:} This work was partially funded
  by support from the Swedish Research Council under project number VR
  2019-03471. The authors declare no competing interests.
\end{abstract}


\section{Introduction}


The Hes1 protein is part of a family of helix-loop-helix repressors
which sustain progenitor cells during development and induce binary
cell differentiation processes \cite{Kageyama2007TheEmbryogenesis}.
Hes1, specifically, plays an important role during neuronal
development and the development of parts of the digestive tract during
embryogenesis, as well as being found to contribute in tumours by ways
of maintaining cancer stem cells and aiding metastasis
\cite{Shimojo2011DynamicCells, Kageyama2007TheEmbryogenesis,
Liu2015Hes1:Resistance}. The exact molecular interactions of these
processes, however, are not yet entirely understood
\cite{Kobayashi2014ExpressionDiseases}, making Hes1 interesting for
mathematical modelling purposes to investigate potential interactions.

To maintain neural progenitor cells, Hes1 oscillates due to a negative
feedback loop between the Hes1 protein and the Hes1 gene
\cite{Hirata2002OscillatoryLoop, Shimojo2011DynamicCells}.
Interactions between the Hes1 negative feedback loop with the
Delta-Notch pathway, a well-conserved developmental pathway
influencing organ development \review{have, in particular, been
observed in developing neural tissue of model organisms such as mice
\cite{Imayoshi2013,Shimojo2016OscillatoryMorphogenesis} and with the
Hes1 homologue Her6 in zebrafish \cite{Soto2020her6}. These
interactions then lead to oscillations throughout a cell population
for a few cycles which dampen over time
\cite{Shimojo2008hes1notch,Phillips2016hes1mir9}, finally resulting in
a sustained ``salt and pepper pattern'' of cells with high and low
levels of Hes1 throughout the population \cite{Imayoshi2013,
Artavanis1999Notch, Kageyama2007TheEmbryogenesis}. Within this
pattern, cells with low Hes1 levels differentiate into neurons via
lateral inhibitions while cells with high levels of Hes1 become
supporting glial cells as observed during, e.g., mouse brain
development \cite{Imayoshi2013,Kageyama2008notch}.}  To allow for the
development of sufficient numbers of each cell type, progenitor cells
need to be maintained at appropriate levels
\cite{Shimojo2011DynamicCells}. Although originally believed to act
like a molecular clock similar to the cell cycle, more recent research
suggests that Hes1 oscillations do not specifically time neural
development during embryogenesis but rather allow cells to stay
undifferentiated for a sufficient amount of time before
differentiation to allow appropriate tissue composition
\cite{Hirata2002OscillatoryLoop, Kobayashi2014ExpressionDiseases}. In
this context, however, all details and functions of Hes1 behaviour
have not yet been understood leading to various mathematical models
seeking to understand and/or explain aspects of these highly complex
molecular interactions. We next review a few modelling frameworks that
have been proposed for the Hes1 system.

%

One type of model that has been explored multiple times is a
relatively simple ordinary differential equation (ODE) model in a
single cell aimed purely at understanding how oscillations can occur
via a negative feedback loop such as in the Hes1 system. Such work has
been done by investigating how Hes1 protein, Hes1 mRNA and an
intermediary factor interact \cite{Hirata2002OscillatoryLoop}, what
role \emph{delay} plays in establishing oscillations
\cite{Barrio2006stochasticdelay, Monk2003OscillatoryDelays,
Jensen2003Hes1DDE, Momiji2008DissectingOscillator}, as well as the
function of dimerisation of the Hes1 protein before it attaches to the
Hes1 promoter \cite{Zeiser2007Hes1Dimerisation}, showing that each of
these models can generate sustained oscillations.

Single cell models have been extended to include more detailed ODE and
partial differential equation (PDE) descriptions.  \review{These
models account for interactions between the Hes1 negative feedback
loop and other cellular pathways, such as the cell cycle
\cite{Pfeuty2015ADynamics}, the cell-internal dynamics including the
accumulation of the microRNA miR-9
\cite{Goodfellow2014mir9,Phillips2016hes1mir9} and the Notch pathway
\cite{Pfeuty2022generalhes,Tiedemann2017ModelingNeurogenesis}, as well
as the spatial distribution of components throughout the cell
\cite{Agrawal2009ComputationalSignaling}. These refinements preserve
oscillatory behaviour while, under specific conditions, allowing for
stable pattern formation in a cell population.}

\review{Additionally, the gene regulatory networks (GRNs) of both the
Hes1 and the related Hes5 proteins have been explored using various
stochastic modelling approaches.  The Hes1 GRN, in particular, has
been effectively captured using the reaction-diffusion master equation
(RDME) approach. By modelling the Hes1 signalling pathway within a
single cell represented by a computational mesh, the RDME method
provides a detailed view of the system by modelling gene regulatory
behaviour on a single interaction basis. This approach not only
captures oscillatory dynamics even in the presence of noise but also
allows for deeper investigation into how nuclear transport and
dimerisation influence the system \cite{Sturrock2013Hes1GRN,
Sturrock2014Dimerisation}.  Similarly, both the Hes1 and the Hes5 GRN
have been modelled using both delay stochastic models such as delay
stochastic simulation algorithm models
\cite{Barrio2006stochasticdelay}, chemical master equation
\cite{Phillips2016hes1mir9} and stochastic differential equation
methods \cite{Manning2019hes5,Hawley2022hes5}, highlighting the
versatility of stochastic approaches in capturing the complex dynamics
of these networks as well as examining how delays contribute to
typical oscillatory behaviour in both single cell
\cite{Agrawal2009ComputationalSignaling,Sturrock2014Dimerisation,Barrio2006stochasticdelay,Phillips2016hes1mir9,Goodfellow2014mir9,Manning2019hes5}
and multicellular
\cite{Pfeuty2022generalhes,Tiedemann2017ModelingNeurogenesis,Hawley2022hes5}
environments.
Although these models describe the Hes1, and related Hes5, pathway in
greater detail, they are also increasingly complex, making it hard to
understand their behaviour analytically.}

Zooming out from the Hes1 specifics and focusing mainly on
developmental patterning in general, the Delta-Notch pathway has been
modelled in multiple ways: From very basic models to determine
patterning behaviour while remaining conducive to analysis
\cite{Collier1996PatternSignalling}, to further extensions including
protrusions and, thus, inducing more extensive patterns than salt and
pepper patterns \cite{Cohen2010, Sprinzak2011, Hadjivasiliou2016,
Engblom2019NDR}. \review{In the two-cell and one-dimensional case,
oscillatory behaviour followed by stable patterning has also been
found by including delay into the Delta-Notch system
\cite{Veflingstad2005hetsteadystate,Momiji2009bistable}.}
Investigations of travelling wavefronts within neurogenesis and the
influence of cell morphology on patterning behaviour
\cite{FormosaJordan2012, Saleh2021CellMorphology} have shown that
patterning is stable across different environments. While some
previous models explicitly include the Hes1-Notch connection
\cite{Agrawal2009ComputationalSignaling, Pfeuty2015ADynamics}, models
purely focusing on the Delta-Notch pathway are also interesting to us
since they have formalised the description of Delta-Notch behaviour
and are amenable to mathematical analysis thanks to a lower model
complexity \cite{Collier1996PatternSignalling}.

We aim at investigating models across different frameworks and start
by modelling the underlying GRN using an ODE system on a grid, based
on the schematics of the biological process. For this we use
parameters drawn from the literature and otherwise determined to the
best of our knowledge. This model captures the oscillatory behaviour
followed by \review{stable expression of Hes1 which indicates} fate
decision \review{into either neurons (at low levels of Hes1) or glial
cells (at high levels)} while keeping the number of modelled molecular
regulators to a minimum. However, this system is still fairly complex
and difficult to analyse so we reduce it to a two-dimensional and even
scalar ODE using quasi-steady state assumptions. In this way, we find
four apparently different but closely related reduced systems which,
although they do not capture the oscillations, allow us to analyse the
timing and behaviour of the fate decision process. We further extend
our ODE model to a spatial stochastic RDME model
\cite{Engblom2009RDME, URDMEpaper} to be able to experience with the
system's stability to intrinsic noise.

We have structured the paper as follows. In \S\ref{sec:models} we
detail the Hes1-Notch signalling model under consideration. Sources of
stochasticity from intrinsic cellular noise as well as spatial effects
are included. We analyse the spectral properties of the model in
\S\ref{sec:analysis}, first assuming a deterministic framework and
using linear stability analysis in space. We investigate the precision
of the analysis as well as its relevance for a more realistic
spatially extended stochastic model. A concluding discussion around
the themes of the paper is found in \S\ref{sec:discussion}.


\section{Models}
\label{sec:models}

Both Hes1 and the Notch pathway are well-preserved pathways and
important during embryonic development \cite{Chen1997, DeLaPompa1997}.
As a fundamental pathway within neurogenesis, it has been extensively
analysed through experiments but a single description of a GRN within
one cell or even the small neighbourhood of its immediate surrounding
cells found from such experiments does not lead into insights into how
single cell interactions lead to population-level behaviour. This
motivates our interest in modelling this behaviour mathematically.

In this section we first describe the underlying biology of the
combined Hes1-Notch GRN in \S\ref{sec:word_model}. Next, we present an
ODE interpretation of this biological pathway on a population of cells
in \S\ref{sec:ode_model}. Finally, in \S\ref{sec:stoch_model} we set
up a stochastic model of the Hes1-Notch pathway, again on the cell
population level, but using the Reaction-Diffusion Master Equation
(RDME) framework.

\subsection{Hes1 Cell-to-Cell Signalling Process}
\label{sec:word_model}

For the modelling, we would ideally like to describe the pathway in a
way that captures its main behaviours while allowing us insight into
mechanistic interactions on a population level through mathematical
analysis and computational simulations. \review{The behaviour we want
  to capture is the behaviour shown in neural progenitor populations
  where Hes1 shows transient oscillations with a period of 2--3 hours
  \cite{Marinopoulou2021hes1neurons,Imayoshi2013,elAzhar2024hes1neurons}
  sometimes extending up to 4 hours
  \cite{elAzhar2024hes1neurons}. This is then followed by a fate
  decision into stationarity, with either high or low Hes1 protein
  levels leading to cells developing into glial cells or neurons,
  respectively \cite{Shimojo2008hes1notch}, which from biological
  considerations has to be rather robust to process and environment
  noise.}

\begin{figure}[t!] 
  \centering
  \includegraphics{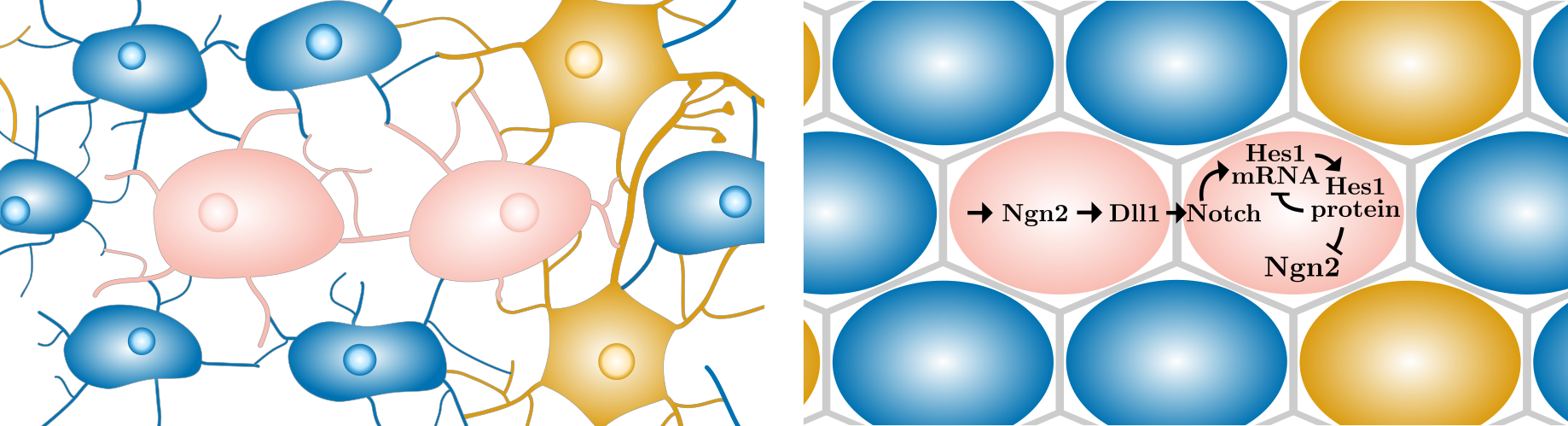}
  \caption{\textit{Left:} Representation of neurons (orange), glial
    cells (blue) and undifferentiated cells (pink) in a developing
    brain. \textit{Right:} Schematics of the Hes1 negative feedback
    loop in two neighbouring cells. The same interactions occur in
    every cell throughout the neural progenitor cell population
    between all neighbouring cells. All arrows ending with an
    arrowhead denote an activation or creation of a constituent while
    arrows ending with a vertical line denote a repression. All
    constituents are also degraded (not shown).}
  \label{fig:hes1_schema}
\end{figure}

We start with the schematic understanding of the underlying biological
processes depicted in \figref{hes1_schema}. Following
\cite{Shimojo2011DynamicCells}, the main molecules involved in the
process of maintaining neural progenitor cells are Neurogenin-2
(Ngn2), Delta-like-1 (Dll1), the Notch receptor as well as Hes1 mRNA
and protein which together interact as indicated in the figure. To
this end we use the notation
\begin{align}
  \label{eq:DOFs}
  D \text{: Dll1, } N \text{: Notch, } M \text{: Hes1 mRNA, }
  P \text{: Hes1 protein, } n \text{: Ngn2,}
\end{align}
for the concentrations of respective constituents in each cell. 

Starting in the left pink cell in \figref{hes1_schema}, Ngn2 is
constitutively produced and induces the production of Dll1, which in
turn is presented on the cell membrane and interacts with the Notch
receptor on the surface of the right cell. The internal part of the
Notch receptor, the Notch intracellular domain, then interacts with
the Hes1 gene promoter to induce the production of Hes1 mRNA which we
summarise here as Notch inducing Hes1 mRNA production. However, we are
aware that both Dll1 and Notch have a bound/inactive form as well as a
free/active form. Both proteins are transmembrane molecules and
signalling occurs via direct contact between the proteins
\cite{Bray2016notch}. This direct contact renders both proteins unable
to function after signalling which becomes relevant during the
mathematical modelling process.

Following the biological pathway further, Hes1 protein is then
produced from the Hes1 mRNA and successively inhibits the production
of new Hes1 mRNA while also repressing the production of the proneural
protein Ngn2. \review{Ultimately, this process causes all key
  components, including Dll1 and Ngn2, to oscillate before stabilising
  into a ``salt and pepper" pattern \cite{Shimojo2008hes1notch,
    Shimojo2016OscillatoryMorphogenesis, Kageyama2008notch}.  In this
  pattern, cells with high levels of Hes1 protein are surrounded by
  cells with low Hes1 levels while cells with low levels of Ngn2 and
  Dll1 tend to have high Hes1 protein levels, while those with high
  Ngn2 and Dll1 levels have low Hes1.}

\subsection{Network ODE Models}
\label{sec:ode_model}

Given the schematic understanding of \figref{hes1_schema}, we start by
proposing an ODE model to describe the Hes1-Notch GRN within a single
cell. In this case we describe Dll1, Notch, Hes1 mRNA, Hes1 protein
and Ngn2 as concentrations $[D, N, M, P, n]$, hence extending purely
Delta-Notch signalling systems such as
\cite{Collier1996PatternSignalling, Cohen2010, FormosaJordan2012} to
also include the Hes1 negative-feedback dynamics \review{similar to
  the model presented in \cite{Pfeuty2022generalhes}, but without
  including delay}.

We let all molecules be degraded at a rate $\mu_i$ with
$i \in \{D,N,M,P,n\}$ and capture the inhibition of Hes1 mRNA as well
as the repression of the production of the proneural protein Ngn2
using the repressor form of Hill functions of the Hes1 protein
\cite{Alon2006Hill}.  At the same time, the activation or production
of each constituent is modelled using $\alpha_i$ according to
individual dynamics of each molecule. These considerations lead to the
system describing the Hes1-Notch GRN in a single cell to be
\begin{align}
  \label{eq:ode}
  \left. \begin{array}{rcl}
           \dot{D} &=& \alpha_D n - \mu_D D, \\
           \dot{N} &=& \alpha_N \langle \Din \rangle - \mu_N N, \\
           \dot{M} &=& \frac{\alpha_M N}{1 + (P/K_M)^k} - \mu_M M,\\
           \dot{P} &=& \alpha_P M - \mu_P P, \\
           \dot{n} &=& \frac{\alpha_n}{1+(P/K_n)^h} - \mu_n n. \\
           \end{array} \right\}
\end{align}
Here, $\langle \Din \rangle := \sum_i w_i D_i$ is the average
time-dependent Dll1 signal a cell receives from its neighbouring cells
$i$ (always normalising the weights $w_i$ to sum to unity). To
determine the cell population behaviour, we apply this ODE system on
each individual node in a network which represents the connectivity
between a population of cells. In this paper we mainly use regular
hexagonal grids, however, other grids can easily be treated in the
same way.

We propose the parameters as given in \tabref{parameters}. Since both
the timings of the entire process with \review{oscillations of periods
  $\sim$ 2--3 hours \cite{Imayoshi2013, Marinopoulou2021hes1neurons}},
as well as most parameter values are available for mouse embryonal
cell lines, our overall calculations are based on these timings for
mouse development. For the degradation rates $\mu_i$ we rely on the
half-lifes for the associated components except for $D$ and $N$ which,
as previously mentioned in \S\ref{sec:word_model}, become inactive
upon contact made by signalling due to proteolytic cleavage of the
Notch receptor. Thus, we assume that $80\%$ of both proteins are used
while $20\%$ are free and can be degraded, \review{i.e., that the
  measured degradation rates are those of the $20\%$ free
  transmembrane proteins, thus, causing the actual degradation rates
  to increase fivefold,} cf.~\tabref{parameters}. One element deciding
system behaviour is the choice of the Hill coefficients $k$ and
$h$. We require both $k,h \in \mathbb{N}^+$ and choose $k = 1$ and
$h = 4$ as these are the minimum values which we have found are
necessary to realistically capture oscillations. \review{Similarly, we
  choose the Hill-function dissociation constants $K_M$ and $K_n$ to
  match the overall system behaviour, as these primarily influence the
  oscillation period and number of oscillations. Since the system is
  underdetermined, we do not account for perturbations in these
  values.} Given degradation rates and with fixed Hill functions, our
activation rates $\alpha_i$ follow by fitting to the relative amounts
of each component as found in \cite{Huang2023paxdb}. The uncertainty
of these activation rates are found by a straightforward Monte Carlo
approach, using the independent perturbations in \tabref{parameters}
and assuming $5\%$ noise for the concentrations. For more information
about this, see Appendix~\ref{app:parameters}. The resulting typical
dynamics of the model are shown in \figref{constituents_comparison}.


   
\begin{table}[t]
  \centering
  \begin{tabular}{ccc}
    \hline
    Parameter& Value (68\% Confidence Interval) & Reference \\
    \hline
    $\alpha_D$ & $0.018 \ (0.016, 0.021)$ [/min] & \textit{This paper} \\ 
    $\alpha_N$ & $6.0 \ (5.3, 6.7)$ [/min] & \\
    $\alpha_M$ & $0.017 \ (0.016, 0.019)$ [/min] & \\
    $\alpha_P$ & $0.14 \ (0.12, 0.16)$ [/min] & \\
    $\alpha_n$ & $0.0049 \ (0.0043, 0.0054)$ [$\mu M$/min] & \\
    \hline
    $\mu_D$ & $\log 2/50 \times 5 \ \log(2) / (45.3,55.2)\times 5$ [/min] & Dll1 half-life in mice  \cite{Shimojo2016OscillatoryMorphogenesis} \\
    $\mu_N $& $\log 2/40 \times 5 \ \log(2)/(36.2,44.2) \times 5$ [/min] & Notch1 half-life in humans \cite{Agrawal2009ComputationalSignaling} \\
    $\mu_M $ & $\log 2/24.1 \ \log 2 /(22.4, 25.8)$ [/min] & Hes1 protein half-life in mice \cite{Hirata2002OscillatoryLoop} \\
    $\mu_P$ & $\log 2/22.3 \ \log 2/(19.2,25.4)$ [/min] & Hes1 mRNA half-life in mice \cite{Hirata2002OscillatoryLoop} \\
    $\mu_n$ & $\log 2/21.9 \ \log 2/(19.7, 24.1)$ [/min] & Ngn2 half-life in Xenopus \cite{Vosper2007RegulationProteolysis} \\
    \hline
    $K_M$ & $\equiv 0.050$ [$\mu M$] & \textit{This paper} \\
    $K_n$ & $\equiv 0.030$ [$\mu M$] & \\
    \hline
    $k$ & $\equiv 1$& \textit{This paper} \\
    $h$ & $\equiv 4$& \\
    \hline
  \end{tabular}
  \caption{Parameters for \eqref{eq:ode}. Values of $\alpha_i$ are
    chosen to give the desired behaviour of constituents relative to
    each other \cite{Huang2023paxdb}, while $\mu_i$ values are based
    on the half-lives of the components of the GRN, mostly in
    mice. For $D$ and $N$, specifically, we make the modelling
    assumption that $80\%$ of each is bound, thus, leading to the
    multiplication by the factor $5$. Where values are available with
    error estimates we use those, while for $\alpha_i$ we fit them to
    all other perturbed parameters, and for $\mu_D$ and $\mu_N$ we
    assume an ad hoc $\pm 10\%$ uncertainty since the value for
    $\mu_N$ is more uncertain (a higher range of values including an
    NICD half-life of $\sim 180$ min
    \cite{Ilagan2011Real-timeReporter} has been reported). To achieve
    the behaviour we desire, the value from
    \cite{Agrawal2009ComputationalSignaling} was used for this
    parameter.}
  \label{tab:parameters}
\end{table}

\begin{figure}[t]
  \centering
  \includegraphics
  {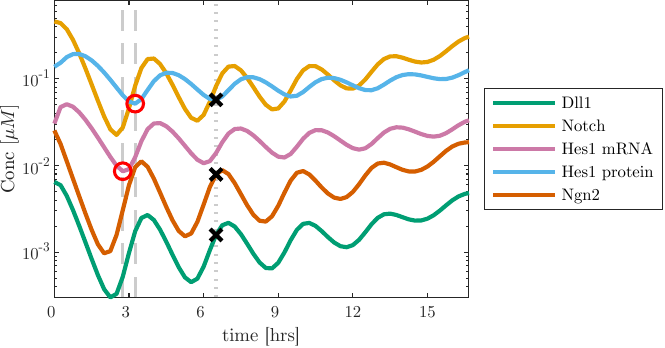}
  \caption{Dynamics averaged over all cells on a $20$-by-$20$ grid of
    hexagonal cells when starting from random initial data. The two
    dashed vertical lines indicate the offset between Hes1 mRNA and
    Hes1 protein expression levels which has been shown previously in
    \cite{Hirata2002OscillatoryLoop}. The offset between the Hes1 mRNA
    and Hes1 protein oscillations between the two markers as shown is
    approximately $30$ minutes. The vertical dotted line shows that we
    approximately capture the inverse oscillations between the Hes1
    protein, and Dll1 and Ngn2 \cite{Shimojo2011DynamicCells}.}
  \label{fig:constituents_comparison}
\end{figure}

For improved ability to analyse the system, we assume quasi-steady
states for three of the five states to find a reduced ODE system.
Depending on the reduction we choose, we find either equations of type
1,
\begin{align}
  \label{eq:red_type1}
  \left. \begin{array}{rclcl}
    \dot{x} &=& \frac{\langle \yin \rangle}{a+x^k} - x &=:& \langle \yin \rangle f(x) - x\\
    \dot{y} &=& v \left( \frac{1}{1 + b x^h} - y \right) &=:& v \left( g(x) - y \right)
  \end{array} \right\},
\end{align}
or of types 2 and 3, respectively,
\begin{align}
  \label{eq:red_type2}
  \left. \begin{array}{rcl}
           \dot{x} &=& y f(x) - x \\
           \dot{y} &=& v \left( \langle g(\xin) \rangle - y \right)
         \end{array} \right\}, \qquad
  \left. \begin{array}{rcl}
          \dot{x} &=& \langle g(\yin) \rangle f(y) - x \\
          \dot{y} &=& v \left( x - y \right)
        \end{array} \right\}, 
\end{align}
where $f$ and $g$ are as for type 1 and where
$\langle g(\xin) \rangle = \sum_i w_i g(x_i)$ is the average of
$g(x_i)$ across the neighbour cells $i$. Overall, there are $10$
possible ways to reduce the original system \eqref{eq:ode} to a
two-dimensional system by making quasi-steady state
assumptions. However, three possible options, \review{those where
  neither of $x$ and $y$ corresponds to $M$ or $P$}, are not readily
reducible since the reduction involves solving Hill equations. This
leaves seven possible alternatives (four of type 1, two of type 2 and
one of type 3) capturing the steady state behaviour of the original
system \eqref{eq:ode}. The different alternatives are summarised in
\tabref{red_alternatives}, and a typical derivation can be found in
Appendix~\ref{app:red}.  For comparison, the behaviour of both the
full model \eqref{eq:ode} and the best fit reduced model
\eqref{eq:red_type1} are shown in \figref{ode_panel}. To note about
the reduced models in \eqref{eq:red_type1} and \eqref{eq:red_type2} is
that they all end up with the same parameters $a$ and $b$,
cf.~\tabref{red_alternatives}, while $v$ varies such that all three
reduced model types behave similarly except for the timing of fate
decision which is determined by $v$.

To further simplify analysis, at points we use a scalar version of our
model. To reach this, we make the further assumption that
$\dot{y} = 0$ in either of the two-dimensional models
\eqref{eq:red_type1}--\eqref{eq:red_type2}. This reduces all three
types into
\begin{equation}
  \label{eq:red_type4}
  \dot{x} = \langle g(\xin) \rangle f(x)-x.
\end{equation}

Our reduced models \eqref{eq:red_type1}--\eqref{eq:red_type2} are
remindful of the Delta-Notch model from
\cite{Collier1996PatternSignalling},
\begin{align}
  \label{eq:collier}
  \left. \begin{array}{rclcl}
    \dot{x} &=& \frac{\langle y_{\text{in}} \rangle ^k}{a + \langle y_{\text{in}} \rangle^k} - x  &=:& F(\langle y_{\text{in}} \rangle ) - x\\
    \dot{y} &=& v \left( \frac{1}{1 + b x^h} - y \right) &=:& v \left( G(x) - y \right)
  \end{array} \right\},
\end{align}
where $x$ describes Notch, $y$ describes Delta, and
$\langle \yin \rangle$ is the average incoming Delta from the
neighbours on the grid.
While our models \eqref{eq:red_type1} and \eqref{eq:red_type2} show
differences in the form of $f(x)$, the order of averaging and Hill
functions, the values of the Hill coefficients $k$ and $h$, as well as
where the model links the incoming signal compared to the Collier
model \eqref{eq:collier}, we can use an analysis similar to the one
proposed in \cite{Collier1996PatternSignalling} to investigate the
behaviour of our system further.

\begin{SCtable}[2][!b]
  \centering
  \begin{tabular}{cccc}
    \hline
    type & $x$ & $y$  & $v$\\
    \hline
    \textbf{1}  & $\mathbf{M}$ & $\mathbf{n}$ & $\mathbf{1.096 \ (0.975, 1.280)}$\\ 
    1 & $P$ & $n$ & $1.014 \ (0.851, 1.100)$\\ 
    1 & $M$ & $D$ & $2.410 \ (2.111, 2.740)$ \\ 
    1 & $P$ & $D$ & $2.230 \ (1.800, 2.627)$ \\ 
    \hline
    2 & $M$ & $N$ & $3.013 \ (2.712, 3.401)$ \\ 
    2 & $P$ & $N$ & $2.788 \ (2.298, 3.280)$ \\ 
    \hline
    3 & $M$ & $P$ & $1.081 \ (0.936, 1.305)$\\ 
    \hline
  \end{tabular}
  \caption{The seven alternative ways to reduce the original system
    \eqref{eq:ode} to \eqref{eq:red_type1} or \eqref{eq:red_type2} via
    quasi-steady state assumptions and the resulting effective
    parameter $v$. The parameters $a$ and $b$ are
    $0.083 \ (0.071, 0.094)$ and
    $1.652 \times 10^5 \ (0.807, 3.217) \times 10^5$ for all
    alternatives ($68\%$ confidence intervals). The version displayed
    in \figref{ode_panel} is indicated in bold.}
  \label{tab:red_alternatives}
\end{SCtable}

\begin{figure}[!t]
  \centering
  \includegraphics{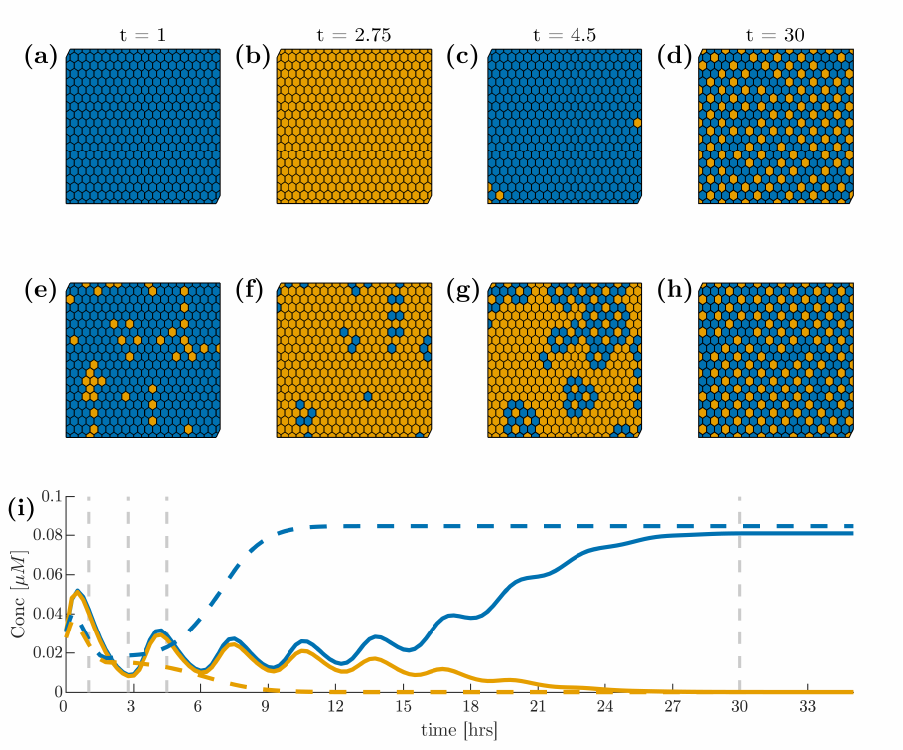}
  \caption{\review{\textit{(a)--(d):}} spatial dynamics of Hes1 mRNA
    in our proposed grid ODE model \eqref{eq:ode} where blue cells are
    above the mean concentration before fate decision and orange cells
    are below this threshold. \review{\textit{(e)--(h):}} Hes1 mRNA in
    the reduced model \eqref{eq:red_type1} on the same
    grid. \review{\textit{(i):}} the average Hes1 mRNA (solid line:
    full ODE model; dashed line: reduced model) over time
    \review{calculated separately over all cells which show high or
      low Hes1 concentrations after fate decision} with blue and
    orange denoting high and low expression, respectively. The
    vertical lines denote the times at which the spatial dynamics are
    shown in the top and middle rows. \review{All simulations shown
      are on a $20 \times 20$ grid with zero boundary
      conditions. Initial conditions are uniform random values scaled
      to the required concentrations as given in
      Appendix~\ref{app:parameters}.} Note that by our
    parameterisation, we find our results in concentrations.}
  \label{fig:ode_panel}
\end{figure}

\subsection{Spatial Stochastic Reaction-transport Model}
\label{sec:stoch_model}

To take intra-cellular noise into account we also consider a
mesoscopic stochastic version of the grid ODE \eqref{eq:ode} as
follows. We represent the individual cells as nodes in a network with
connectivity given by an underlying mesh discretisation. Consider a
single cell first, with time-dependent state vector
$X(t) \in \Intdom_{+}^{d}$ counting at time $t$ the number of
constituents (or species) in each of $d$ compartments. We may
generally prescribe $R$ Markovian reactions in the form of Poissonian
state transitions $X \mapsto X+\stoich_{r}$ by
\begin{align}
  \label{eq:prop}
  \Prob\left[X(t+dt) = x+\stoich_{r}| \; X(t) = x\right] &=
                                                           w_{r}(x) \, dt+o(dt),
\end{align}
for $r = 1\ldots R$ with $w_r(x)$ the $r$th transition intensity (or
propensity), and $\stoich \in \Intdom^{d \times R}$ the stoichiometric
matrix. The evolution of the $i$th species can then be described by
the Poisson representation \cite{Markovappr}
\begin{align}
  \label{eq:Poissrepr}
  X_i(t) &= X_i(0)+\sum_{r = 1}^{R} \stoich_{ri} \Pi_{r}
           \left(  \int_{0}^{t} w_{r}(X(s)) \, ds \right),
\end{align}
with unit-rate and independent Poisson processes
$(\Pi_{r})_{r = 1}^{R}$.

\begin{SCfigure}
  \includegraphics{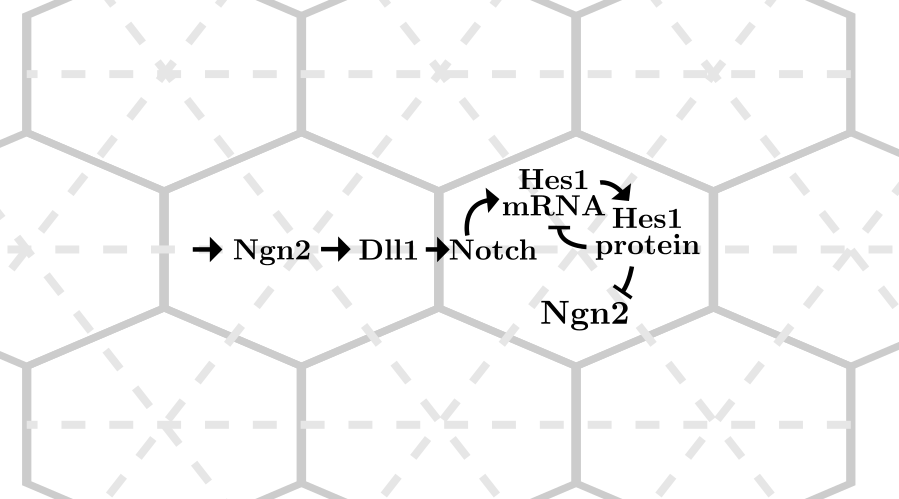}
  \caption{The schematics as implemented on a hexagonal grid (solid
    lines) using the RDME model.  The dashed lines show the
    triangulation on which the hexagonal grid is built in the URDME
    framework.}
  \label{fig:triangulation}
\end{SCfigure}
  
In the present case we identify the following reactions:
\begin{align}
  \label{eq:RDME1}
  \left. \begin{array}{rcl}
           n  & \xrightarrow{ \alpha_Dn } &  n+D	\\
           N  & \xrightarrow{ \alpha_M N/(1+(P/(K_M \VoxVol))^k) } &  N+M	\\
           M  & \xrightarrow{ \alpha_P M } &  M+P	\\
           \emptyset  & \xrightarrow{ \alpha_n\VoxVol/(1+(P/(K_n\VoxVol))^h) } &  n
         \end{array} \right\} \qquad
  \left. \begin{array}{rcl}
           D  & \xrightarrow{ \mu_D D } &  \emptyset	\\
           N  & \xrightarrow{ \mu_N N } &  \emptyset	\\
           M  & \xrightarrow{ \mu_M M } &  \emptyset	\\
           P  & \xrightarrow{ \mu_P P } &  \emptyset	\\
           n  & \xrightarrow{ \mu_n n } &  \emptyset
         \end{array} \right\}
\end{align}
where $\VoxVol$ is the volume of each voxel. The production of Notch,
as initiated by the Dll1 signal, is yet to be described.

We next consider a population of cells in $K$ nodes or voxels
$(\VoxVol_k)_{k = 1}^K$ and a time-dependent state
$X \in \Intdom_{+}^{d \times K}$, with $X_{ik}(t)$ the number of
constituents of the $i$th species in the $k$th voxel. The general
dynamics \eqref{eq:Poissrepr} now becomes
\begin{align}
  \label{eq:RDMEPoissrepr}
  X_{ik}(t) = X_{ik}(0) &+
  \sum_{r = 1}^{R} \stoich_{ri} \Pi_{rk}  
  \left( \int_{0}^{t} \VoxVol_k u_{r}(\VoxVol_k^{-1}X_{\cdot ,k}(s)) \, ds \right) \\
  \nonumber
  &-\sum_{k = 1}^{J} \Pi_{ijkl}' 
  \left( \int_{0}^{t} q_{ijkl}X_{ik}(s) \, ds \right)
  +\sum_{k = 1}^{J} \Pi_{jilk}'
  \left( \int_{0}^{t} q_{jilk}X_{jl}(s) \, ds \right),
\end{align}
where $q_{ijkl}$ is the rate per unit of time for species $i$ in the
$k$th voxel to transfer into species $j$ in the $l$th voxel, and where
$(\Pi_{\cdot},\Pi'_{\cdot})$ is an appropriately extended set of
independent unit-rate Poisson processes. This general linear transfer
process is not standard as it allows for species to change their type
while transporting, but it is appropriate here since it is exactly
this effect we are interested in. Note also that in
\eqref{eq:RDMEPoissrepr}, the propensities $(u_r)$ are independent of
the voxel volume $\VoxVol_k$. Using this formalism we may augment
\eqref{eq:RDME1} with
\begin{align}
  \label{eq:RDME2}
  \left. \begin{array}{rrl}
           D  & \xrightarrow{ \alpha_N D }&  D+D^{\text{in}} \\
           D_k^{\text{in}} & \xrightarrow{ \alpha_N q_{kl} D_k^{\text{in}}}&  N_l
         \end{array} \right\}
\end{align}
that is, a Dll1 signal in voxel $k$ sequentially transforms into a
diffusing pseudo species $D^{\text{in}}$, which then diffuses
\emph{into} a Notch signal in voxel $l$ at rate $\alpha_N q_{kl}$,
where $q_{kl}$ is the proportion of Dll1 used for the signal between
these two voxels (for example, $q_{kl} \equiv 1/6$ on a hexagonal mesh
with $k$ and $l$ neighbouring voxels).

The model so described can readily be implemented across a given
triangulation of space using URDME \cite{URDMEpaper} and simulated
using the supported NSM-solver with a triangulation as illustrated in
\figref{triangulation}. Sample simulations are reported in
\figref{urdme_panel}.

\begin{figure}[H]
  \centering
  \includegraphics{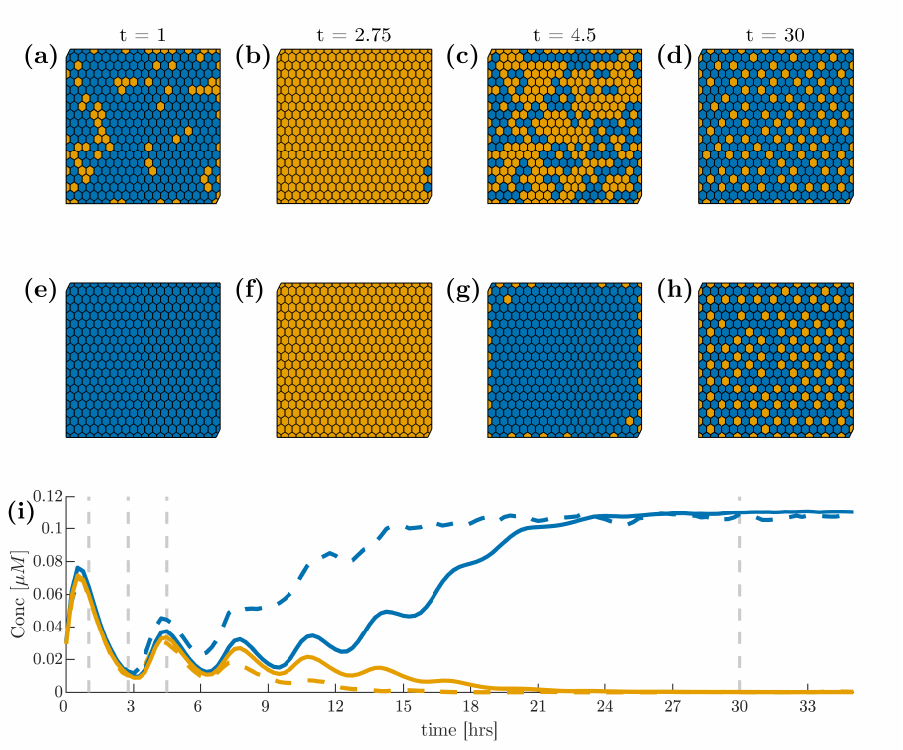}
  \caption{\review{\textit{(a)--(d):}} spatial dynamics of Hes1 mRNA
    in our RDME model \eqref{eq:RDME1} and \eqref{eq:RDME2} choosing
    the volume of each voxel to be $1 \mu m^3$, representing a rather
    high noise levels, and using the same colour scheme as in
    \figref{ode_panel}. \review{\textit{(e)--(h):}} Hes1 mRNA in the
    RDME model with voxel volume $50 \mu m^3$, i.e., less levels of
    noise. \review{\textit{(i):}} the average Hes1 mRNA at low volume,
    $1 \mu m^3$, (dashed line) and high volume, $50 \mu m^3$, (solid
    line) over time where the horizontal lines denote the times at
    which the spatial dynamics are shown in the top and middle
    rows. Blue and orange, again, denote cells with high and low
    expression \review{and boundary and initial conditions are chosen
      as previously only this time initial conditions are in number of
      molecules}.  Based on a mouse embryonal stem cell volume of
    approximately $50 \mu m^3$ (size based on \cite{Pillarisetti2009,
      Wang2011} assuming spherical cells) and a mean number of $8104$
    molecules per cell \cite{Ho2018abundance}, we find our results in
    $\mu M$.}
\label{fig:urdme_panel}
\end{figure}


\section{Analysis and Results}
\label{sec:analysis}

We next analyse the properties of the system \eqref{eq:ode}. Existence
and qualitative behaviour of fate decision in a two-cell 1D periodic
system in the reduced model \eqref{eq:red_type1}--\eqref{eq:red_type4}
is investigated in \S\ref{sec:steady_2} and in the full model
\eqref{eq:ode} in \S\ref{sec:steady_5}. We then examine the behaviour
of the system \eqref{eq:red_type1} on a regular hexagonal grid in
\S\ref{sec:spatial_periodic}, and we finally quantitatively compare
the patterning differences between the ODE \eqref{eq:ode} and RDME
models \eqref{eq:RDME1}--\eqref{eq:RDME2} in \S\ref{sec:ode_rdme}.

\subsection{The Reduced Stationary Solutions}
\label{sec:steady_2}

At stationary solutions to \eqref{eq:ode}, the quasi-stationary
arguments used to arrive at the reduced systems
\eqref{eq:red_type1}--\eqref{eq:red_type4} are valid and so we target
these models initially. We first consider the homogeneous steady state
where, by ``homogeneous'' we simply mean that all cells have identical
states. We pick the scalar reduced model \eqref{eq:red_type4}, i.e.,
\begin{align}
  \label{eq:ND0}
    \dot{x} &= g(\langle \xin \rangle) f(x)-x,
\end{align}
where $f, g$ are as in \eqref{eq:red_type1}. Looking for a homogeneous
steady state where $\langle \xin \rangle = x$, we define
$\varphi(x) := g(x)f(x)$ and equivalently search for fixed points
satisfying $\varphi(x) = x$. Since $0 < \varphi(0)$, $\varphi(1) < 1$,
and since $f$, $g$, and, hence, also $\varphi$ are all decreasing
functions there is a unique root $\bar{x}_0$ in $(0,1)$,
cf.~\figref{fixpnt}. In conclusion,
\begin{proposition}
  \label{prop:hom_exists}
  There is a unique stationary point $\bar{x}_0 \in (0,1)$ for the
  homogeneous problem \eqref{eq:ND0}. By extension this unique
  solution also applies to the homogeneous version of the full system
  \eqref{eq:ode}.
\end{proposition}


Since we want to show that our system undergoes fate decision into a
\emph{non-homogeneous} solution, we next investigate the stability
properties of the homogeneous steady state in the simplest
one-dimensional setting consisting of two cells with a
periodic boundary condition.

\begin{proposition}
  \label{prop:hom_unstable1D}
  The homogeneous stationary solution in the reduced system
  \eqref{eq:red_type1} is unstable in a system with two cells under a
  periodic boundary condition if and only if
  \begin{align}
    \label{eq:cond0}
    f(\bar{x}_0)g'(\bar{x}_0)-f'(\bar{x}_0)g(\bar{x}_0) &< -1,
  \end{align}
  for $\bar{x}_0$ the homogeneous stationary solution.
\end{proposition}

\begin{proof}
  The two-cell periodic system reads
  \begin{equation}
    \label{eq:twocell}
    \left. \begin{aligned}
        \dot{x_1} &= g(x_2) f(x_1) - x_1 \\
        \dot{x_2} &= g(x_1) f(x_2) - x_2
      \end{aligned} \right\}.
  \end{equation}
  We assume small perturbations about the homogeneous steady state and
  introduce the change of variables
  \begin{equation}
    \label{eq:changeofvariables}
    \begin{aligned}
      \sigma &= \frac{x_1 + x_2}{2}, \quad &\delta &= \frac{x_1 - x_2}{2},
    \end{aligned}
  \end{equation}
  where we consider the perturbation $\delta$ small. Expanding the
  system around the homogeneous stationary solution, the equations
  decouple and we find the governing equation
  \begin{equation}
    \label{eq:twocell_linearised}
    \dot{\delta} = \left(
      f'(\sigma)g(\sigma)-f(\sigma)g'(\sigma) -1 \right) \delta.
  \end{equation}
  Letting $\sigma = \bar{x}_0$ we obtain condition \eqref{eq:cond0}.
\end{proof}

This result holds for parameter $(a,b) > 0$ which holds for all
reductions to any of the reduced systems. That the previous result
remains true for the full system \eqref{eq:ode} is more involved to
show so we defer this to the next section \S\ref{sec:steady_5}.

\begin{figure}[t]
  \centering
  \includegraphics{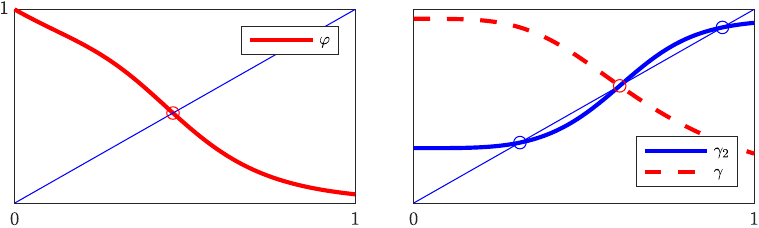}
  \caption{Fix point arguments. \textit{Left:} the unique homogeneous
    stationary state is the fix point
    $\bar{x}_0 = \varphi(\bar{x}_0)$. \textit{Right:} if
    $\gamma_2'(\bar{x}_0) > 1$, then there are cyclic
    (non-homogeneous) solutions $\bar{x}_1 < \bar{x}_0 < \bar{x}_2$.}
  \label{fig:fixpnt}
\end{figure}

We next consider the existence of a non-homogeneous steady state. We
again assume a 2-cell periodic set up and thus look for stationary
solutions to \eqref{eq:twocell}.

\begin{proposition}
  \label{prop:non-hom_exists}
  Under condition \eqref{eq:cond0} there exists a non-homogeneous
  stationary state for the 2-cell periodic problem \eqref{eq:twocell}.
  By extension this solution also applies to the corresponding
  periodic version of \eqref{eq:ode}.
\end{proposition}

\begin{proof}
  From \eqref{eq:twocell} we have the stationary relation
  \begin{equation*}
    h(x_1) := \frac{x_1}{f(x_1)} = g(x_2), \mbox{ or }
    x_1 = h^{-1}(g(x_2)) =: \gamma(x_2).
  \end{equation*}
  For positive arguments, the function $h$ is increasing, hence
  $h^{-1}$ is increasing too, and with $g$ decreasing, $\gamma$ is
  therefore a decreasing function. One readily shows that
  $\gamma(0) > 0$ and $\gamma(1) < 1$ which together forms a second
  proof of the existence of the unique fix point $\bar{x}_0$ for the
  homogeneous stationary state. However, we are rather interested in
  cyclic solutions, i.e., for which
  $x = (\gamma \circ \gamma)(x) =: \gamma_2(x)$, since these
  correspond to alternating (patterned) solutions in the 2-cell
  problem. It is easy to see that $\gamma_2(0) > 0$ and
  $\gamma_2(1) < 1$ and since $\gamma_2(\bar{x}_0) = \bar{x}_0$ we
  find two additional solutions $\bar{x}_1 < \bar{x}_0 < \bar{x}_2$
  under the condition that $\gamma_2'(\bar{x}_0) > 1$,
  cf.~\figref{fixpnt}~(\textit{right}). We get
  \begin{align}
    \label{eq:cond0_}
    \frac{d}{d\xi} \gamma(\gamma(\xi)) 
    \vert_{\xi=\bar{x}_0} > 1 &\iff \gamma'(\gamma(\xi))\gamma'(\xi)
    \vert_{\xi=\bar{x}_0} = \gamma'(\bar{x}_0)^2 > 1 \iff \gamma'(\bar{x}_0) < -1.
  \end{align}
  We find via implicit differentiation and using
  $\bar{x}_0 = \gamma(\bar{x}_0)$ that
  \begin{align*}
    \gamma'(\bar{x}_0) &= \frac{g'(\bar{x}_0)}{h'\left( \gamma(\bar{x}_0)\right)}
                         = \frac{f(\bar{x}_0)
                         g'(\bar{x}_0)}{1-\bar{x}_0 \, f'(\bar{x}_0)/f(\bar{x}_0)}
                         = \frac{f(\bar{x}_0) g'(\bar{x}_0)}{1-g(\bar{x}_0)f'(\bar{x}_0)},
  \end{align*}
  revealing that, in fact, \eqref{eq:cond0_} is equivalent to
  condition \eqref{eq:cond0}.
\end{proof}

One cannot rule out the existence of more than one set of non-homogeneous
solutions. To select a specific one, we pick the one pair
$(\bar{x}_1,\bar{x}_2)$ which is the furthest away from
$\bar{x}_0$. By inspection this solution also satisfies
\begin{align}
  \label{eq:prop_}
  \gamma_2'(\bar{x}_1) = \gamma_2'(\bar{x}_2) = \gamma'(\bar{x}_1)
  \gamma'(\bar{x}_2) < 1,
\end{align}
cf.~\figref{fixpnt}~(\textit{right}). Interestingly, this property
guarantees stability of this solution as we next demonstrate.

\begin{proposition}
  \label{prop:non-hom_stable1}
  The non-homogeneous solution of
  Proposition~\ref{prop:non-hom_exists} is stable whenever it exists.
\end{proposition}

\begin{proof}
  The Jacobian around the non-homogeneous solution has the
  characteristic polynomial
  \begin{align*}
    p(\lambda) &= (\lambda - f'_1g_2+1)(\lambda - f'_2g_1+1)-
                 f_1 f_2 g_1'g_2',
  \end{align*}
  where $f_1 = f(\bar{x}_1)$ and similarly for $f'_1$, $g_1$, $g'_2$,
  etc.  By inspection all coefficients are positive except for the 0th
  order term. By Descarte's rule of sign there is a positive real
  eigenvalue if an only if this term is negative, that is, the
  non-homogeneous stationary solution is stable if and only if
  \begin{align}
    \label{eq:cond1_}
    0 &< f'_1f'_2g_1g_2 - f_1f_2g_1'g_2' - f'_1g_2 - f'_2g_1 + 1.
  \end{align}
  For the function $\gamma$ introduced in the proof of
  Proposition~\ref{prop:non-hom_exists}, we have
  \begin{align*}
    \gamma'(\bar{x}_1) &=
        \frac{g'_1}{h'\left( \gamma(\bar{x}_1)\right)}
        = \frac{f_2 g'_1}{1-\bar{x}_2 \, f'_2/f_2}
        = \frac{f_2 g'_1}{1-f'_2g_1},
  \end{align*}
  and similarly for $\gamma'(\bar{x}_2)$. Hence, from rearranging the
  property \eqref{eq:prop_} we find
  \begin{align*}
    (1-f'_2g_1)(1-f'_1g_2) &> f_1 f_2 g'_1 g'_2,
  \end{align*}
  which is equivalent to condition \eqref{eq:cond1_}.
\end{proof}

%

So far we have shown that there always exists a unique homogeneous
stationary solution. For the 2-cell periodic problem and under
condition \eqref{eq:cond0}, this solution is unstable and there is
then another non-homogeneous solution which \emph{is}
stable. \figref{fac_bifurcation} illustrates this behaviour along a
certain selected path in parameter space for the full model
\eqref{eq:ode}. We next proceed to show that as suggested by this
graphic, the results indeed hold for the full model as well.

\begin{SCfigure}
  \includegraphics{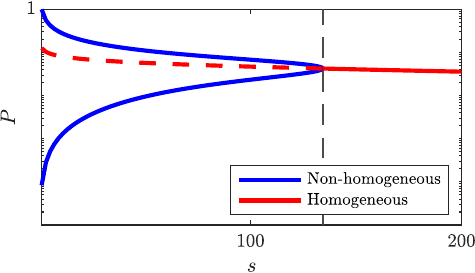}
  \caption{The non-homogeneous and homogeneous stationary states
    \review{of the Hes1 protein} $P$ (log-scale), respectively, as a
    function of a scaling $s$, which acts upon the parameter
    $\alpha_N$, by scaling $\alpha_N \mapsto s^{-1}\alpha_N$. The
    homogeneous solution always exists, but is unstable to the left of
    the dashed line which indicates the smallest value of $s$ for
    which \eqref{eq:cond0} and \eqref{eq:cond1} are true.}
  \label{fig:fac_bifurcation}
\end{SCfigure}

\subsection{Extension to the Full Model}
\label{sec:steady_5}

To understand in what way the reduced models capture the stability
properties of the full model, we need to describe how they are related
at sufficient detail. Let a general ODE have the form $\dot{x} = F(x)$
and assume that the state has been split according to $x = [y; \; z]$,
that is,
\begin{align}
  \begin{bmatrix}\dot{y} \\ \dot{z} \end{bmatrix}
  &= \begin{bmatrix} F_y(y,z)\\
              F_z(y,z) \end{bmatrix}.
\end{align}
The reduced model for $z$ is obtained by assuming that
$\dot{y} \approx 0$ and such that, given $z$, $y$ can be uniquely
solved for
\begin{align}
  0 &= F_y(y,z) \iff y = G(z).
\end{align}
The reduced model is then simply
\begin{align}
  \dot{z} &= F_z(G(z),z),
\end{align}
and the reduced model's Jacobian is given by
\begin{align}
  \label{eq:J_z}
  J_z &= \partial_y F_z G'(z)+\partial_z F_z =
        -\partial_y F_z [\partial_y F_y]^{-1}\partial_zF_y+\partial_z F_z.
\end{align}
By contrast, the full Jacobian reads
\begin{align}
  J &= \begin{bmatrix}
                \partial_y F_y & \partial_z F_y \\
                \partial_y F_z & \partial_z F_z
              \end{bmatrix},
\end{align}
and by a block decomposition \cite{MatrixAnalysis} the determinant is
given by
\begin{align}
  \label{eq:det_J}
  \det(J) &= \det(\partial_y F_y ) \times \det \left( \partial_z
            F_z-\partial_y F_z [\partial_y F_y]^{-1}\partial_zF_y \right).
\end{align}

In general, both Jacobians $J$ and $J_z$ depend on a parameter vector
$\theta$, say, such that we can write $J = J(\theta)$ and equivalently
for $J_z$. Since the determinant of the negative Jacobian is the 0th
order term of the characteristic polynomial, we formulate the
following lemma by comparing \eqref{eq:J_z} and \eqref{eq:det_J}:
\begin{lemma}
  \label{lem:sign}
  Let $p_x(\lambda) \equiv \det(\lambda I-J)$ be the characteristic
  polynomial for the full Jacobian and equivalently define
  $p_z(\lambda) \equiv \det(\lambda I-J_z)$. Suppose that for some
  parameter $\theta$, all coefficients are positive except for
  possibly the 0th order term $p_x(0)$. Suppose also that the order
  reduction is definite in the sense that $\partial_y F_y$ is either
  positive or negative definite for all considered parameters
  $\theta$. Then, as a function of $\theta$, $p_x(0)$ switches sign
  simultaneously with $p_z(0)$ and in fact,
  $p_x(0) = \det(-\partial_y F_y) \times p_z(0)$.
\end{lemma}

The main use of the lemma is in conjunction with Descarte's rule of
sign as it allows one to conclude that the spectrum of $J$ switches
from stable to unstable at points for which $J_z$ is singular. The
expressions for these points are typically simpler to obtain than for
the full system. However, one still has to show that the full
characteristic polynomial has positive terms of higher order than 0.

\begin{proposition}
  \label{prop:hom_unstable1D_full}
  Let $P_0$ be the homogeneous stationary solution for state $P$ of
  the full model \eqref{eq:ode}. This solution is unstable for the
  2-cell periodic problem if and only if
  \begin{align}
    \label{eq:cond1}
    -\frac{h (P_0/K_n)^h}{1+(P_0/K_n)^h}+
    k P_0 (P_0/K_M)^k (1+(P_0/K_n)^h) \times R &< -1,
  \end{align}
  where $R \equiv \prod_i \alpha_i/\mu_i$, $i \in \{ D,N,M,P,n\}$.
\end{proposition}

Under the reduction \eqref{eq:red_type1} (cf.~Appendix~\ref{app:red})
we have that
\begin{align}
  \label{eq:red1}
  a &= (K_M/R)^{k/(k+1)}, \\
  b &= \left( K_M^{k/(k + 1)}/K_n \times R^{1/(k+1)} \right)^h, \\
  \label{eq:red3}
  P_0 &= R^{1/(k+1)} \times (\alpha_P/\mu_P)^{2k/(k+1)} K_M^{-k/(k+1)}
        \times \bar{x}_0,
\end{align}
where we recall that $\bar{x}_0 \in (0,1)$ is the homogeneous
stationary solution for the reduced model as in
Proposition~\ref{prop:hom_exists}.

\begin{proof}
  After the same type of change of variables as in
  \eqref{eq:changeofvariables} and linearising around small
  perturbations, we obtain a relatively sparse linear time-dependent
  system. The characteristic polynomial $p(\lambda)$ can therefore be
  obtained via iterated cofactor expansions. Writing
  $f_k(x) = 1/(1+x^k)$ and similarly for $f_h$, we find
  \begin{align*}
    p(\lambda)
    &= (\lambda + \mu_D)(\lambda + \mu_N)(\lambda + \mu_n)
      \left[(\lambda + \mu_M)(\lambda + \mu_P) - N_0 \alpha_M \alpha_P
      f'_k(P_0/K_M)/K_M\right] \\
    &\hphantom{=}
      +\alpha_D \alpha_N \alpha_M \alpha_P \alpha_n
      f'_h(P_0/K_n) f_k(P_0/K_M)/K_n.
  \end{align*}
  By inspection all coefficients of the polynomial are positive except
  for possibly the constant term. We verify that
  $\det(-\partial_y F_y) = \mu_N \mu_M \mu_D \mu_n > 0$ in the
  notation of Lemma~\ref{lem:sign} and so it follows that the
  stability condition \eqref{eq:cond0} is preserved by the state
  reduction. Using the relations \eqref{eq:red1}--\eqref{eq:red3} and
  the fix point relation $\bar{x}_0 = f(\bar{x}_0)g(\bar{x}_0)$ we
  find that \eqref{eq:cond0} is equivalent to \eqref{eq:cond1}.
\end{proof}

It remains to show that the non-homogeneous solution also shares its
stability properties with the reduced model.

\begin{proposition}
  \label{prop:non-hom_stable1D_full}
  The non-homogeneous solution of
  Proposition~\ref{prop:non-hom_exists} is stable
  for the full model \eqref{eq:ode} whenever it exists.
\end{proposition}

\begin{proof}
  This time we linearise around the non-homogeneous solution and
  obtain a 10-by-10 Jacobian. Luckily the Jacobian is rather sparse
  such that its characteristic polynomial can be expanded into
  \begin{align*}
    p(\lambda) &= (\lambda + \mu_D)^2(\lambda + \mu_N)^2(\lambda +
                 \mu_n)^2 \times \\
               &\phantom{=} \left[(\lambda + \mu_M)(\lambda + \mu_P) -
                 N_1\alpha_M\alpha_Pf'_k(P_1/K_M)/K_M\right] \times \\
                 &\phantom{=} \left[(\lambda + \mu_M)(\lambda + \mu_P) -
                 N_2\alpha_M\alpha_Pf'_k(P_2/K_M)/K_M \right] \\
               &\phantom{=} -
                 (\alpha_D\alpha_N\alpha_M\alpha_P\alpha_n/K_n)^2
                 f'_h(P_1/K_n)f'_h(P_2/K_n)f_k(P_1/K_M)f_k(P_2/K_M).
  \end{align*}
  All coefficients of the polynomial are positive except for possibly
  the constant term. The reduction map is verified to be positive
  definite and so we conclude that the stability condition
  \eqref{eq:cond1} again controls the stability also of the
  non-homogeneous solution.
\end{proof}

\subsection{Patterning on Regular Hexagonal Tilings}
\label{sec:spatial_periodic}

Next we are interested in analysing the patterning that occurs when a
non-homogeneous steady state is reached. From a biological
perspective, a ``random'' pattern or a chaotic non-stationary
behaviour is implausible in a highly regulated pathway. As we have
shown in Propositions \ref{prop:hom_unstable1D} and
\ref{prop:hom_unstable1D_full}, the homogeneous steady state is
unstable in both the reduced and the full models under conditions
\eqref{eq:cond0} and \eqref{eq:cond1}, assuming a two-cell system with
periodic couplings. When the homogeneous steady state \emph{is}
unstable, the heterogeneous steady state exists and is stable
(Propositions \ref{prop:non-hom_exists}, \ref{prop:non-hom_stable1},
\ref{prop:non-hom_stable1D_full}). Again, this holds for the simple
case of a two-cell system with periodic couplings, but we are
nevertheless led to believe that a regular periodic pattern will
eventually result from the model.

On a regular hexagonal tiling, there are a multitude of regularly
periodic patterns that can occur, yet not every such periodic pattern
is \emph{uniform} or \emph{vertex-transitive}. In the case of the
non-uniform pattern illustrated in~\figref{regular_pattern}, for
example, we require two `sub-types' of white cells: those that border
black cells and those that do not. We argue that for symmetry reasons,
assuming a vertex-transitive pattern is a reasonable assumption for
further investigations of our model. There exist three such uniform
colourings on a regular hexagonal tiling \cite{tilings},
see~\figref{hex_tilings}, one of which describes the homogeneous case.

\begin{SCfigure}[][t]
    \includegraphics{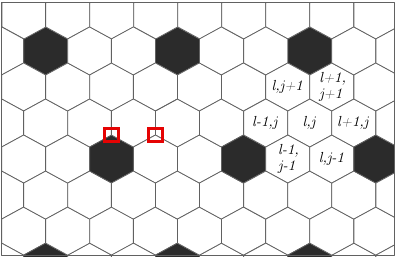}
    \caption{A regularly periodic pattern on a hexagonal grid with
      period $3$ in each lattice direction. The labelling scheme shown
      follows \cite{Collier1996PatternSignalling}. The red squares
      highlight two vertices which make the pattern non-transitive:
      the left borders 2 white and one black cell, while the right
      borders 3 white cells.}
    \label{fig:regular_pattern}
\end{SCfigure}

\begin{figure}[b]
    \centering
    \includegraphics{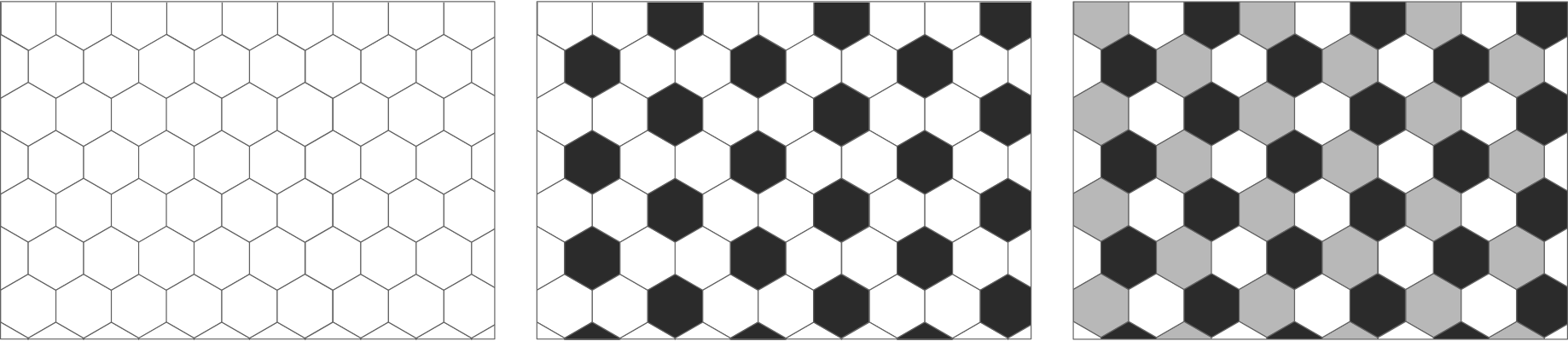}
    \caption{The three patterns on a regular hexagonal tiling which
      show vertex transitive behaviour, i.e., every vertex has the
      same neighbours. Note that the right pattern is identical to the
      middle pattern if two colours are the same.}
    \label{fig:hex_tilings}
\end{figure}

\begin{proposition}
  Under periodic boundary conditions the homogeneous steady state of
  the reduced system \eqref{eq:red_type4} is conditionally unstable in
  both one and two dimensions. Conditions for instability are the
  previous \eqref{eq:cond0}, while in 2D the condition is found in
  \eqref{eq:cond1per} below.
\end{proposition}

\begin{proof}
  On a one-dimensional lattice with $j = 1,2,\ldots,N$, we find from
  \eqref{eq:red_type4} the governing equations
  \begin{align*}
    \dot{x}_j &= \frac{g(x_{j-1})+g(x_{j+1})}{2} 
                f(x_j)-x_j. \\
    \intertext{Linearising around the homogeneous solution and writing
    $x_j = \bar{x}_0+\delta_j$ we find}
    \dot{\delta}_j &= g'f \frac{\delta_{j-1}+\delta_{j+1}}{2}
                     +gf' \delta_j-\delta_j.
  \end{align*}
  Inserting the Fourier representation
  \begin{align*}
    \delta_j &= \sum_{s=1}^N \xi_s \exp \left( \frac{2 \pi i s j }{N}
               \right), \\
    \intertext{we get}
    \dot{\xi}_s &= \left( g'f\cos\left( \frac{2\pi}{N} s \right)+gf'-1\right)\xi_s.
  \end{align*}
  By inspection the most unstable case occurs for $s/N = 1/2$ which
  is then equivalent to condition \eqref{eq:cond0}.

  On a two-dimensional, regular hexagonal lattice with
  $j=1,2,\ldots,M$ and $l=1,2,\ldots,N$, we similarly find the
  governing equations
  \begin{align*}
    \dot{x}_{lj} &= \frac{\sum g(x_{l\pm 1,j\pm 1})}{6}
                   f(x_{lj})-x_{lj}, \\
    \intertext{where the sum involves the 6 lattice neighbours,
    cf.~\figref{regular_pattern}. Linearising around the homogeneous
    solution, we find}
    \dot{\delta}_{lj} &= g'f \frac{\sum \delta_{l\pm 1,j\pm 1}}{6} 
                        +gf' \delta_{lj}-\delta_{lj}.
  \end{align*}
  Again making use of the Fourier representation,
  \begin{align*}
    \delta_{lj} &= \sum_{r=1}^M \sum_{s=1}^N \xi_{rs} \exp \left( 
                  \frac{2 \pi irl}{M} + \frac{2 \pi isj}{N} \right), \\
    \intertext{this becomes}
    \dot{\xi}_{rs} &= \left( g'f A+
                     gf'-1\right)\xi_{rs}, \\
    \intertext{where}
    3A &\equiv \cos\left( \frac{2\pi}{N} s \right)+
                     \cos\left( \frac{2\pi}{M} r \right)+
                     \cos\left( \frac{2\pi}{N} s + \frac{2\pi}{M} r
                     \right).
  \end{align*}
  The most unstable case occurs for $s = r = N/3$, assuming $N = M$
  and divisibility by 3. The implied condition for instability is then
  \begin{align}
    \label{eq:cond1per}
    f(\bar{x}_0)g'(\bar{x}_0)/2-f'(\bar{x}_0)g(\bar{x}_0) &< -1.
  \end{align}
\end{proof}
It is tempting to draw the conclusion that the corresponding unstable
frequency is also the resulting pattern: with period $N/3$ this would
indeed imply the middle pattern in \figref{hex_tilings} which is also
what we observe from numerical experiments. However, the analysis only
reveals the most unstable modes around the homogeneous solution and
does not predict the eventual end-fate.

From numerical experiments we consistently find that the typical
stationary pattern generally matches that
of~\figref{hex_tilings}~(\textit{middle}), with black/white
corresponding to, respectively, low/high Hes1 protein concentrations.
\review{Since the stationary state on the hexagonal mesh only consists
  of two distinct states it seems intuitive to attempt to analyse the
  two-dimensional situation by looking at the two-cell model coupled
  according to}
\begin{align}
  \label{eq:W2}
  W_2 &= \begin{bmatrix}
    0 & 1 \\
    1/2 & 1/2
  \end{bmatrix}, \\
  \intertext{\review{i.e., as observed each black cell only has white
  neighbours while white cells have on average three white and three
  black neighbours. Thus, we consider the generic coupled model}}
  \label{eq:generic_type4}
  \dot{x}_i &= \sum_j W_{ij} g(x_j) \times f(x_i)-x_i,
\end{align}
for $i = 1,2$ and $W = W_2$. However, this immediate two-dimensional
extension of the two-cell periodic one-dimensional case gives
incorrect results. This case supports non-homogeneous \emph{stable}
solutions which are close to the homogeneous one but which are never
observed in larger simulations.

A better generalisation \review{of the observed behaviour on the
  hexagonal mesh} is rather \emph{three} cells, that is, the smallest
integer multiple of three as suggested by the previous Fourier
analysis. Namely, we take the generic model \eqref{eq:generic_type4}
with $W = W_3$,
\begin{align}
  \label{eq:W3}
  W_3 &= \begin{bmatrix}
    0 & 1/2 & 1/2 \\
    1/2 & 0 & 1/2 \\
    1/2 & 1/2 & 0
  \end{bmatrix},
\end{align}                    
and $i = 1,\ldots ,3$ \review{to represent the vertex-transitive model
  matching \figref{hex_tilings}~(\textit{right})}. Consider first the
labels ``low/medium/high'' concentrations, say, at a stationary state
$[x_1,x_2,x_3]$. Since the non-homogeneous stationary solution
consists of either low or high concentration we will make the
identification that ``medium'' corresponds to ``high'' concentration,
i.e., $x_2 = x_3$, mimicking the
way~\figref{hex_tilings}~(\textit{right}) can be transformed
into~\figref{hex_tilings}~(\textit{middle}). Conveniently, the
stationary states can now be found by considering the simpler
extension \eqref{eq:W2}--\eqref{eq:generic_type4} since the stationary
relations are the same. Following the approach in the proof of
Proposition~\ref{prop:non-hom_exists} we have
\begin{align}
  \left. \begin{aligned}
      x_1 &= g(x_2)f(x_1) \\
      x_2 &= \frac{g(x_1)+g(x_2)}{2} f(x_2)
    \end{aligned} \right\} \quad
            \left. \begin{aligned}
                h(x_1) &= g(x_2) \\
                H(x_2) &= \frac{g(x_1)}{2}
              \end{aligned} \right\} \quad
                         \left. \begin{aligned}
                             x_1 &= \gamma(x_2) \\
                             x_2 &= \Gamma(x_1)
                           \end{aligned} \right\},
\end{align}
where $H$ and $\Gamma$ are defined in analogy with $h$ and
$\gamma$. Alternating (cyclic) solutions are now found from
\begin{align}
  \label{eg:Gammadef}
  x_1 &= (\gamma \circ \Gamma)(x_1) =: \Gamma_2(x_1),
\end{align}
and as before a sufficient condition for existence would be
$\Gamma_2'(\bar{x}_0) > 1$. Unfortunately this approach fails due to
the existence of multiple cyclic solutions as the numerical experiment
in \figref{3cell_bifurcation} explains. Except for in singular points
there are now \emph{two} pairs of non-homogeneous solutions and the
crossing at the homogeneous solution generally satisfies
$\Gamma_2'(\bar{x}_0) < 1$. Inspired by this graphical motivation, we
instead proceed by assuming that non-homogeneous solutions exist for
\emph{some} parameter combination and we attempt to find points for
which all such non-homogeneous solutions vanish.

\begin{proposition}
  The boundary for existence of non-homogeneous solutions is defined
  by
  \begin{align}
    \label{eq:cond1_exists}
    \frac{g'_1f_2}{2-[f'_2g_2+f_2g'_2]-g_1f'_2} \times
    \frac{f_1g'_2}{1-f'_1g_2} &= 1,
  \end{align}
  where $f_1 = f(\bar{x}_1)$ and similarly for $f'_1$, $g_1$, $g'_2$,
  etc.
\end{proposition}

\begin{proof}
  By the graphical motivation in
  \figref{3cell_bifurcation}~\textit{(right)} we search for a double
  root at the non-homogenous solution. That is, for which
  $\Gamma_2'(\bar{x}_1) = \Gamma'(\bar{x}_1)\gamma'(\bar{x}_2) =
  1$. We find the derivatives through implicit differentiation leading
  to the two factors in the expression \eqref{eq:cond1_exists}.
\end{proof}

\begin{figure}[t]
  \centering
  \includegraphics{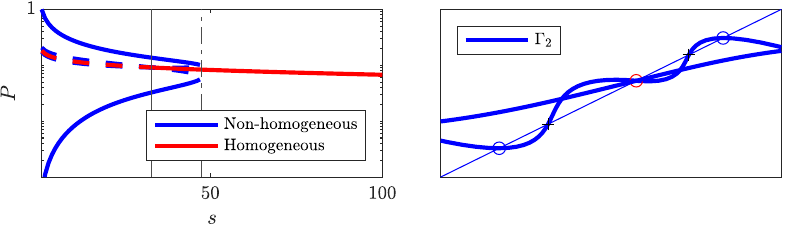}
  \caption{Stationary solutions under vanishing
    feedback. \textit{Left:} bifurcation diagram \review{of the Hes1
      protein $P$} as a function of scaling $s$;
    $\alpha_N \mapsto s^{-1}\alpha_N$. Below the boundary point
    defined by \eqref{eq:cond1_exists} (\textit{dash-dot vertical})
    there exist two pairs of non-homogeneous solutions, but only one
    pair is stable. Below the critical value \eqref{eq:cond1per}
    (\textit{solid vertical}) the homogeneous solution is also
    unstable. \textit{Solid/dashed} indicates stable/unstable
    solutions, respectively. \textit{Right:} illustration of the fix
    point problem \eqref{eg:Gammadef} when the parameter $s$ is just
    above and below the limit for existence of non-homogeneous
    stationary solutions. When $s$ increases, the pair of unstable
    solutions (\textit{plus-signs}) approaches the pair of stable
    solutions (\textit{circles}) until they collapse into a double
    root (condition \eqref{eq:cond1_exists}), after which only the
    homogeneous solution remains.}
  \label{fig:3cell_bifurcation}
\end{figure}

\begin{SCfigure}
  \includegraphics{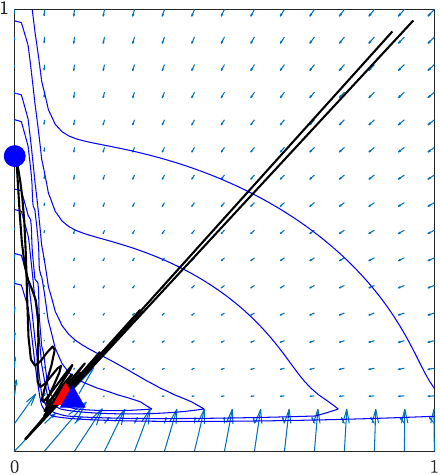}
  \caption{Phase-plot of the three-cell problem in the plane which
    contains all three stationary points. \textit{Circle:} stable
    non-homogeneous solution, \textit{triangles:} unstable solutions
    with red the homogeneous one. \textit{Level curves} according to
    the Euclidean norm of the right-hand side and \textit{arrows}
    denote the direction of the flow. Two sample trajectories starting
    from the upper right corner are also shown, with the oscillations
    visible before the fate decision.}
  \label{fig:3cell_phase_plot}
\end{SCfigure}

To sum up, under the proposed parameters from~\tabref{parameters} and
under weakened feedback $\alpha_N \to 0+$, the model undergoes a
transition where the typical checkerboard patterning is lost. The
two-dimensional generalisation into three cells, as given by
\eqref{eq:generic_type4}--\eqref{eq:W3}, displays the same stability
of post fate decision patterning as consistently observed for the full
model when simulated over a grid of multiple connected cells. As an
illustration, the overall typical dynamics of the three-cell system is
summarised in \figref{3cell_phase_plot}

\subsection{System's Size Convergence}
\label{sec:ode_rdme}

Finally we want to at least qualitatively investigate the RDME model
\eqref{eq:RDME1}--\eqref{eq:RDME2}. As with the ODE model, we first
investigate its stability to scaling in the single parameter
$\alpha_N$. Unlike the ODE model, one is now forced to use a
statistical procedure to estimate when the non-homogeneous solutions
are lost. This is complicated by the fact that the level of noise is
rather large around the transition points, We find in
\figref{umod_bifurcation} that while the bifurcation behaviour is not
as clear as in \figref{3cell_bifurcation}, there is a gradual change
of system behaviour approximately around the critical value(s) of
scaling $s$, where the system behaviour changes from non-homogeneous
(patterned) into homogeneous. In \figref{umod_bifurcation} we select
the top 25\% and bottom 25\%, respectively, of the stationary Hes1
protein levels, and plot their means. We use the diameter of each
sample as a measure of the spread and judge if the distribution is
bimodal or not. Other statistical procedures yield slightly varying
results but overall, we find that the RDME behaviour over a mesh is
fairly well predicted by the solutions and critical points of the
three cell problem.

\begin{SCfigure}[][b]
  \centering
  \includegraphics{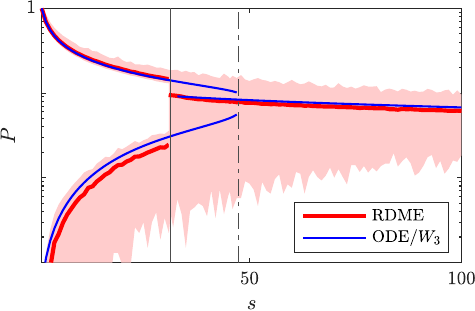}
  \caption{Same scaling as in \figref{3cell_bifurcation}, but
    instead solving the RDME-model on a two-dimensional hex-grid and
    estimating the mean behaviour of high and low protein
    concentrations among the cells. The stable solutions of the ODE
    for the three cell problem from \figref{3cell_bifurcation} are
    included for reference.}
  \label{fig:umod_bifurcation}
\end{SCfigure}

Finally, we try to measure the quality of the fate decision in the
RDME model. For this purpose, recall the coupling matrix $W_2(i,j)$ as
defined in \eqref{eq:W2}, where $i$, $j$ are the two different fates
low/high Hes1 protein levels, respectively. For a perfect pattern with
two final expression modes (\figref{hex_tilings} (\emph{middle})), the
coupling $W_2$ is as given in \eqref{eq:W2}. To evaluate the behaviour
of the RDME model, we seek to estimate the effective coupling matrix
from observations. By running multiple independent simulations and
splitting the cells of the resulting stationary process into low/high
Hes1 expression, we can count how often the specific coupling
high-high occurs out of all possible couplings, that is, which
corresponds to $W_2(2,2)$, \review{the ``patterning coefficient''
  $p$}. We next treat these counts as independent Bernoulli trials and
hence the set-up can be practically approached as a statistical
estimation problem for the single Bernoulli parameter
\review{$p$}. The results are summarised in \figref{patterning} and
indicate that even at relatively large levels of noise (corresponding
to small cell volume), the patterning is quite close to the perfect
one. For example, for mouse embryonal stem cells with a size of about
$50 \mu m^3$, the GRN patterning behaviour is comparably stable to
intrinsic cellular noise, with \review{the patterning coefficient}
$\hat{p} = 0.50 (0.43,0.57)$ (95\% CI).

\begin{SCfigure}
    \centering
    \includegraphics{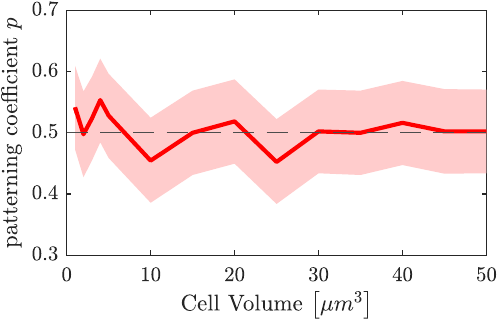}
    \caption{A measure of connectivity between high and low protein
      cells (perfect patterning corresponds to \review{patterning
        coefficient} $p = 1/2$) for a sequence of volumes in the RDME
      model \eqref{eq:RDME1}. The numerical problem is treated as a
      statistical estimation problem for which the confidence
      intervals are indicated as a shaded area. See text for details.}
    \label{fig:patterning}
\end{SCfigure}


\section{Discussion}
\label{sec:discussion}

Ultimately, we developed a first principle ODE model of the Hes1
pathway and its direct interactions with the Notch pathway, capturing
oscillations followed by cell differentiation. We chose parameters
based on biological data as much as feasible. By reducing our initial
ODE model to lower dimensions, we were able to analyse the
differentiation process. Furthermore, we extended the ODE model into a
spatial stochastic RDME model to investigate the system's robustness
to intrinsic noise. We found that the transient oscillatory and the
differentiation processes observed in the ODE model were well
preserved in the RDME model. In this way we have found multiple
interlinked ways to model this signalling process, allowing us to
investigate multiple aspects of the behaviour of the Hes1-Notch
pathway. Linking modelling frameworks in this way forced us to think
hard about parameter scaling and parameterisation issues, particularly
so in relation to the scarce availability of \review{quantitative}
experimental data.


Our models capture essential aspects of the Hes1-Notch pathway,
although they do not replicate every observed behaviour. We observed a
few dampened oscillations followed by stable patterning both in the
presence and absence of noise. However, the number of oscillations,
their stability as well as their length is limited by parameter choice
and the interpretation of the states included in the model.
As stated in \cite{Kobayashi2014ExpressionDiseases}, ``[t]he molecular
mechanism by which cells exhibit different (oscillatory vs.\
sustained) expression modes of Hes1 is still unknown''. Nevertheless,
our results suggest that these distinct expression modes may be
intrinsic to the GRN, and thus independent of external regulation.

\review{Despite these promising findings, our models have some
  limitations.  Our model's oscillations, for example, are longer and
  more dampened than found in experimental data
  \cite{Hirata2002OscillatoryLoop, Imayoshi2013,
    Marinopoulou2021hes1neurons}. However, more recent results suggest
  greater variability in oscillation periods in neural precursor cells
  with periods between 2--4 hours \cite{elAzhar2024hes1neurons} a
  range consistent with our results.  Additionally, the final
  patterning behaviour of the RDME model (cf.\
  \figref{umod_bifurcation} and \ref{fig:patterning}) remains largely
  stable under noise, closely resembling the deterministic case. This
  suggests that noise does not crucially affect patterning behaviour
  in agreement with previous findings that stochastic behaviour in the
  Hes1 signalling pathway is stable to noise
  \cite{Phillips2016hes1mir9}. Such stability is desirable to such a
  developmentally significant differentiation pathway.  While both the
  ODE and RDME behaviour (cf.~\figref{ode_panel} and
  \figref{urdme_panel}) show similarities in oscillatory and fate
  decision behaviours, noise in the stochastic model allows for
  earlier fate decision in individual cells. As a result, oscillations
  appear less pronounced at the population level, again consistent
  with previous findings showing earlier differentiation in smaller
  systems \cite{Phillips2016hes1mir9}.}


Our focus was on the Hes1-Notch pathway isolated from other cellular
and signalling processes though we have made simplifications. We did
not account for processes such as the dimerisation of the Hes1 protein
before it induces mRNA production
\cite{Kageyama2007TheEmbryogenesis}\review{, its interactions with the
  micro RNA miR-9 \cite{Goodfellow2014mir9,Phillips2016hes1mir9}} or
other interactions with Notch effectors such as Mash1
\cite{Shimojo2011DynamicCells}.  Additionally, the Hes1 pathway
interacts with multiple other pathways, such as the cell cycle
\cite{Pfeuty2015ADynamics}, RBP-J and Jagged
\cite{Kageyama2007TheEmbryogenesis}, as well as the JAK-STAT pathway
\cite{Kobayashi2014ExpressionDiseases}. These interactions could
potentially stabilise the oscillations, which are currently severely
dampened in our model. Furthermore, the addition of, for example,
extra states or a delay to simulate this behaviour in an ODE system
has been shown to allow for more stable oscillations
\cite{Goodwin1965model,Griffith1968goodwin, Monk2003OscillatoryDelays,
  Jensen2003Hes1DDE, Zeiser2007Hes1Dimerisation}.  However, such an
extension of our model would further complicate mathematical analysis,
which conflicts with our aim of balancing analysability and model
complexity.

A major challenge in modelling cellular signalling pathways is the
limited availability of \review{quantitative} data, leading to many
models relying heavily or even exclusively on ad hoc parameter values
chosen to fit expected model behaviour
\cite{Gunawardena2010parameters}. As the aim of many Hes1 and
Delta-Notch models is to hypothesise about how specific molecular
interactions induce cellular and population-level behaviour, the
significance of results not based on biologically relevant parameters
have to be considered with caution. \review{While qualitative time
  series data is more widely available, e.g.,
  \cite{elAzhar2024hes1neurons, Imayoshi2013}, fairly advanced
  parameter inference is needed to make effective use of this type of
  data, adding significant complexity.} In our work, this data
sparsity has particularly affected the parameters
$\alpha_i, K_M, K_n, k$ and $h$. Nevertheless, we have scaled our
model behaviour, and thus our parameters, to align with expected
biological concentrations for the different constituents to improve
the significance of our results. \review{Since our primary focus is on
  model development and analysis, we leave more advanced parameter
  inference for future work.}  As a point in favour of our approach,
the corresponding RDME model successfully simulates the Hes1-pathway
down to the resolution of single species.

There are several promising ways to extend the work presented
here. Extending the model to incorporate interactions with additional
pathways, such as the JAK-STAT pathway or others, could further
clarify the dynamics of Hes1 expression. Similarly, adding detail to
the Hes1-Notch pathway itself \review{to, for example, replace the
  phenomenological Hill functions with mechanistic models of gene
  regulation} could improve model behaviour and give further insights
into its oscillatory and sustained expression modes.  Moreover, as the
change between these two modes of expression happens during embryonal
development to allow for sufficient numbers of neurons and glial cells
to develop \cite{Kobayashi2014ExpressionDiseases}, considering the
stability of this signalling process in a growing population becomes
exceedingly relevant. Additionally, Hes1 is an important factor in the
development of other tissues as well as different cancer types
\cite{Kobayashi2014ExpressionDiseases} so investigations of the
differences between Hes1 interactions in these different tissues could
be of interest as well.

To sum up, we have used different modelling approaches to capture the
essential behaviours of the Hes1-Notch signalling pathway during
neuronal development. By balancing simplicity and analytical
tractability, we have constructed models that are both amenable to
mathematical analysis and biologically meaningful. Using these models
we capture the intrinsic expression of both oscillations and final
patterning of this pathway while basing parameters on experimental
data as much as possible. \review{While there are improvements or
  different modelling emphases that can be chosen, our work promotes
  further theoretical understanding of the oscillatory and
  differentiation processes of this critical signalling pathway and
  the trade-offs between analysability and feasibility of the
  computational modelling.}

\subsection{Availability and reproducibility}
\label{subsec:reproducibility}

The computational results can be reproduced with release 1.4 of the
URDME open-source simulation framework \cite{URDMEpaper}, available
for download at \url{www.urdme.org}. Refer to the Hes1 directory and
the associated \texttt{README.md} in the DLCM workflow.


\printbibliography[title={References}]


\clearpage
\appendix

\section{Parameterisation}
\label{app:parameters}

From \cite{Ho2018abundance}, we know that the average amount of Hes1
protein in a \textit{S.~cerevisiae} cell is $8104$ molecules. Since
the yeast cells are roughly similar in size to mouse embryonal stem
cells, we assume a size of $50 \mu m^3$ for both the
\textit{S.~cerevisiae} cells and the mouse embryonal stem cells. This
gives a wanted Hes1 protein concentration of $0.269 \mu M$. For the
concentrations of the remaining constituents, we use data from the
PaxDB database \cite{Huang2023paxdb} using the values for the
integrated whole organism of the mouse for Dll1, Notch1 and Hes1 while
we use the Ngn2 value for the integrated whole organism of humans as
this value was not available for mice. Since the values in PaxDB are
given in parts per million, we use them to determine the
concentrations of the constituents relative to each other. The values
used and their scalings relative to Hes1 protein are shown in
\tabref{par_scale}.  Finally, we use results from \cite{Yu2006mrna} to
scale the Hes1 mRNA concentration. The authors show that there are
$4.2 \pm 0.5$ protein molecules per mRNA molecule which gives a Hes1
mRNA concentration of $0.0613 \mu M$.

Using this information, we are able to scale $\alpha_i, K_M$ and $K_n$
such that the mean stationary behaviour is equal to those wanted
concentrations in~\tabref{par_scale}. After experimenting with this
set-up, we fixed the \review{Hill-function dissociation constants $K_M$ and $K_n$} at suitable values
such that $\alpha_i$ are now uniquely defined as a function of all the
other parameters.
%
%
To find the distributions and confidence intervals of the activation
rates, we perturb all degradation rates $\mu_i$ assuming they are
log-normally distributed with a $68\%$ confidence interval as
described in \tabref{parameters}. We also perturb the desired
concentrations in \tabref{par_scale}, assuming an ad hoc level of
$5\%$ relative uncertainty as well as them being log-normally
distributed. This way, we find perturbed activation rates $\alpha_i$
by fitting model behaviour to our previous requirements. We then fit
lognormal distributions to the thus sampled $\alpha_i$ to find their
confidence intervals as given in \tabref{parameters}.

Finally, the distributions for all parameters of the full ODE model
in~\tabref{parameters} induce distributions of the reduced parameters,
which we again find by fitting lognormal distributions to the
perturbed reduced parameters, cf.\ \tabref{red_alternatives}.

\begin{table}[t]
  \centering
  \begin{tabular}{llr}
        \hline
    Molecule & PaxDB Value & Scaled Concentration \\
    \hline
    Dll1 & $0.03$ ppm [\textit{M.~musculus}] & $0.0135 \ \mu M$ \\
        Notch1 & $2.2$ ppm [\textit{M.~musculus}] & $0.925 \ \mu M$ \\
    Hes1 mRNA & -- & $0.0613 \ \mu M$ \\
    Hes1 protein & $0.639$ ppm [\textit{M.~musculus}] & $0.269 \ \mu M$ \\
    Ngn2 & $0.12$ ppm [\textit{H.~sapiens}] & $0.0505 \ \mu M$ \\
    \hline
  \end{tabular}
  \caption{Concentration information for the four protein constituents
    of the model from \cite{Huang2023paxdb} and scaled to fit relative
    to Hes1 protein concentration.}
  \label{tab:par_scale}
\end{table}

\section{Dimensional Reduction}
\label{app:red}

To reduce our ODE model \eqref{eq:ode} down to two equations, we rely
on quasi-steady state assumptions. To show a typical reduction, we
choose the third alternative as shown in \tabref{red_alternatives},
i.e.\ a type 1 reduction. Reductions of type 2 and 3 work equivalently
and give the same values for $a$ and $b$.  The substitutions in this
case are
\begin{equation*}
  \quad M \leftrightarrow x, \quad D \leftrightarrow y.
\end{equation*}

We assume that the following variables are approximately stationary,
\begin{equation*}
    \begin{aligned}
        \dot{N} \approx 0 \implies N &= \frac{\alpha_N}{\mu_N} \langle \Din \rangle,\\
        \dot{P} \approx 0 \implies P &= \frac{\alpha_P}{\mu_P} M, \\
        \dot{n} \approx 0 \implies n &= \frac{\alpha_n}{\mu_n} \frac{1}{1+(P/K_n)^k},
    \end{aligned}
\end{equation*}
where $\langle \Din \rangle := \sum_i w_i D_i$.

This gives
\begin{align*}
  &\left. \begin{array}{rcl}
            \dot{M} &=& c_1 \frac{\langle \Din \rangle}{1 + c_2 M^k} - \mu_M M \\
            \dot{D} &=& c_3 \frac{1}{1 + c_4 M^h} - \mu_D D
          \end{array} \right\}
                        \intertext{with}
                        c_1 &= \frac{\alpha_N \alpha_M}{\mu_N}, 
                              \quad c_2 = \left( \frac{\alpha_P}{\mu_P K_M} \right)^k, 
                              \quad c_3 = \frac{\alpha_D \alpha_n}{\mu_n}, 
                              \quad c_4 = \left( \frac{\alpha_P}{\mu_P K_n} \right)^h.
\end{align*}

Using the scaling
\begin{equation*}
    x = \frac{M}{M_0}, \quad y = \frac{D}{D_0}, \quad \Tilde{t} = \frac{t}{\tau},
\end{equation*}
we find the non-dimensionalised equations
\begin{align*}
  &\left. \begin{array}{rcl}
            \dot{x} &=& \frac{\langle \yin \rangle}{a + x^k} - x \\
            \dot{y} &=& v \left( \frac{1}{1 + b x^h} - y \right)
          \end{array} \right\}
\end{align*}
where
\begin{equation*}
  M_0 = \left( \frac{c_1 D_0}{\mu_M c_2} \right)^{\frac{1}{k+1}},
  \quad D_0 = \frac{c_3}{\mu_D},
  \quad \tau = \frac{1}{\mu_M},
\end{equation*}
and
\begin{equation*}
  a^{-1} = c_2 M_0^k,
  \quad b = c_4 M_0^h,
  \quad v = \frac{\mu_D}{\mu_M}.
\end{equation*}

\end{document}